\documentclass[runningheads,envcountsame]{llncs}
\usepackage{fullpage}
\usepackage[utf8]{inputenc}
\usepackage{amsmath}
\usepackage{enumitem}
\usepackage{setspace}
\usepackage{graphicx}
\usepackage{microtype}
\usepackage[colorlinks]{hyperref}
\usepackage[capitalize]{cleveref}

\usepackage{xcolor}
\definecolor{teigiIro}{HTML}{5700B5}
\newcommand*{\teigi}[1]{{\color{teigiIro}\emph{#1}}} 

\usepackage{xcolor}
\newcommand{\XSays}[3]{{\color{#2}
      {$\rule[-0.12cm]{0.2in}{0.5cm}$\fbox{\tt
            #1:} }\itshape #3
      \marginpar{\color{#2}\tt #1}\def\comment{#3}\def\empty{}\ifx\comment\empty\else
      {$\rule[0.1cm]{0.3in}{0.1cm}$\fbox{\tt
            end}$\rule[0.1cm]{0.3in}{0.1cm}$} \fi
   }}

\usepackage[linesnumbered,ruled,vlined]{algorithm2e}
	\SetKwComment{Comment}{$\triangleright$\ }{}
\SetKwProg{Function}{function}{}{}

\usepackage[draft]{fixme}
\fxsetup{theme=color,mode=multiuser}
\FXRegisterAuthor{hb}{ahb}{HB}
\FXRegisterAuthor{tk}{atk}{TK}

\fxsetface{env}{}

\newcommand*{\LCS}{\textrm{LCS}}

\newcommand{\UnaryOperator}[2][]{\ifx&#1&\ensuremath{\mathop{}\mathopen{}#2\mathopen{}}\else \ensuremath{\mathop{}\mathopen{}#2\mathopen{}(#1)}\fi }
\newcommand{\mathup}[1]{#1}

\newcommand{\Oh}[1]{\UnaryOperator[#1]{\mathcal{O}}}

\newcommand{\Ot}[1]{\UnaryOperator[#1]{\mathup{\Theta}}}

\crefname{algocf}{alg.}{algs.}
\Crefname{algocf}{Algorithm}{Algorithms}

\title{Computing Longest (Common) Lyndon Subsequences} 

\author{
        Hideo~Bannai\inst{1}\orcidID{0000-0002-6856-5185}
        \and
        Tomohiro~I\inst{2}\orcidID{0000-0001-9106-6192}
        \and
        Tomasz~Kociumaka\inst{3}\orcidID{0000-0002-2477-1702}
        \and
        Dominik~K\"{o}ppl\inst{1}\orcidID{0000-0002-8721-4444}
        \and
        Simon~J.~Puglisi\inst{4}\orcidID{0000-0001-7668-7636}
}
\authorrunning{H. Bannai, T. I, T. Kociumaka, D. K\"{o}ppl, and S. J. Puglisi.}

\institute{
        M\&D Data Science Center, Tokyo Medical and Dental University, Japan\\
        \email{\{hdbn,koeppl\}.dsc@tmd.ac.jp}
        \and
        Department of Artificial Intelligence, Kyushu Institute of Technology, Japan\\
        \email{tomohiro@ai.kyutech.ac.jp}
        \and
        University of California, Berkeley, United States,
        \email{kociumaka@berkeley.edu}
        \and
        Department of Computer Science, Helsinki University, Finland\\
        \email{simon.puglisi@helsinki.fi}
}

\begin{document}
\maketitle
\begin{abstract}
  Given a string~$T$ with length~$n$ whose characters are drawn from an ordered alphabet of size $\sigma$, 
  its longest Lyndon subsequence is a longest subsequence of~$T$ that is a Lyndon word.
  We propose algorithms for finding such a subsequence in \Oh{n^3} time with \Oh{n} space,
  or \emph{online} in \Oh{n^3 \sigma} space and time.
  Our first result can be extended to find the longest common Lyndon subsequence of two strings of length~$n$ in \Oh{n^4 \sigma} time using \Oh{n^3} space.
\end{abstract}

\section{Introduction}
A recent theme in the study of combinatorics on words has been the generalization of regularity properties from substrings to subsequences.
For example, given a string~$T$ over an ordered alphabet, the longest increasing subsequence problem is to find the longest subsequence of increasing symbols in~$T$~\cite{GBRobinson1938,Schensted61}.
Several variants of this problem have been proposed~\cite{Knuth70,elmasry10longest}.
These problems generalize to the task of finding such a subsequence that is not only present in one string, but common in two given strings~\cite{kutz11faster,ta21computing,he18longest}, which can also be viewed as a specialization of the longest common subsequence problem~\cite{wagner74correction,kiyomi21longest,hirschberg77algorithms}.

More recently, the problem of computing the longest square word that is a subsequence~\cite{kosowski04efficient},
the longest palindrome that is a subsequence~\cite{chowdhury14computing,inenaga18hardness},
the lexicographically smallest absent subsequence~\cite{kosche21absent},
and longest rollercoasters~\cite{DBLP:conf/gd/BiedlCDJL17,DBLP:conf/stacs/GawrychowskiMS19,DBLP:journals/siamdm/BiedlBCLMNS19,DBLP:conf/spire/FujitaNIBT21} have been considered.

Here, we focus on subsequences that are Lyndon, i.e.,
strings that are lexicographically smaller than any of its non-empty proper suffixes~\cite{lyndon54}. Lyndon words are objects of longstanding combinatorial interest, and have also proved to be useful algorithmic tools in various contexts (see, e.g.,~\cite{BIINTT17}).
The longest Lyndon \emph{substring} of a string is the longest factor of 
the Lyndon factorization of the string~\cite{chen58lyndon}, and can be computed in linear time~\cite{duval83lyndon}.
The longest Lyndon \emph{subsequence} of a unary string is just one letter,
which is also the only Lyndon subsequence of a unary string.
A (naive) solution to find the longest Lyndon subsequence is to enumerate all distinct Lyndon subsequences, and pick the longest one. However, the number of distinct Lyndon subsequences can be as large as $2^n$ considering a string of increasing numbers $T = 1 \cdots n$.
In fact, there are no bounds known (except when $\sigma=1$) that bring this number in a polynomial relation with the text length~$n$ and the alphabet size~$\sigma$~\cite{hirakawa21counting},
and thus deriving the longest Lyndon subsequence from all distinct Lyndon subsequences can be infeasible.
In this paper, we focus on the algorithmic aspects of computing this {longest} Lyndon subsequence in polynomial time without the need to consider all Lyndon subsequences.
In detail, we study the problems of computing
\begin{enumerate}
\item the lexicographically smallest (common) subsequence for each length online, cf.~\cref{secLexicographicallySmallest}, and
\item the longest subsequence that is Lyndon, cf.~\cref{secLongestLyndon}, with two variations considering the computation as online, or the restriction that this subsequence has to be common among to given strings.
\end{enumerate}
The first problem serves as an appetizer. 
Although the notions of \emph{Lyndon} and \emph{lexicographically smallest} share common traits, 
our solutions to the two problems are independent,
but we will reuse some tools for the online computation.

\section{Preliminaries}
Let $\Sigma$ denote a totally ordered set of symbols called the alphabet.
An element of $\Sigma^*$ is called a string.
Given a string $S \in \Sigma^*$, we denote its length with $|S|$,
its $i$-th symbol with $S[i]$ for $i \in [1..|S|]$.
Further, we write $S[i..j] = S[i]\cdots S[j]$.
A \teigi{subsequence} of a string~$S$ with length~$\ell$ is a string $S[i_1] \cdots S[i_\ell]$ with $i_1 < \ldots < i_\ell$.

Let $\bot$ be the empty string.
We stipulate that $\bot$ is lexicographically larger than every string of $\Sigma^+$.
For a string $S$, appending $\bot$ to~$S$ yields~$S$.

A string $S \in \Sigma^*$ is a \teigi{Lyndon word}~\cite{lyndon54} if $S$ is lexicographically smaller than all its non-empty proper suffixes.
Equivalently, a string $S$ is a Lyndon word if and only if it is smaller than all its proper cyclic rotations.

The algorithms we present in the following may apply techniques limited to integer alphabets.
However, since the final space and running times are not better than $\Oh{n}$ space and $\Oh{n \lg n}$ time, respectively, 
we can reduce the alphabet of $T$ to an integer alphabet by sorting the characters in $T$ with a comparison based sorting algorithm taking \Oh{n \lg n} time and \Oh{n} space, 
removing duplicate characters, and finally assigning each distinct character a unique rank within $[1..n]$.
Hence, we assume in the following that $T$ has an alphabet of size $\sigma \le n$.

\section{Lexicographically Smallest Subsequence}\label{secLexicographicallySmallest}

As a starter, we propose a solution for the following related problem:
Compute the lexicographically smallest subsequence of~$T$ for each length $\ell \in [1..n]$ online.

\subsection{Dynamic Programming Approach}\label{secFirstApproachLexicographicallySmallest}

The idea is to apply dynamic programming dependent on the length~$\ell$ and the length of the prefix $T[1..i]$
in which we compute the lexicographically smallest subsequence of length~$\ell$.
We show that the lexicographically smallest subsequence of $T[1..i]$ length~$\ell$, denoted by $D[i,\ell]$ is
$D[i-1,\ell]$ or $D[i-1,\ell-1]T[i]$,
where $D[0,\cdot] = D[\cdot,0] = \bot$ is the empty word.
See \Cref{algoLexSmallest} for a pseudo code.

\begin{algorithm}[t]
  \DontPrintSemicolon{}
  $D[0,1] \gets \bot$ \;
  \For(\Comment*[f]{Initialize $D[\cdot,1]$}){$i = 1$ to $n$}{$D[i,1] \gets \min_{j \in [1..i]} T[j] = \min(D[i-1,1],T[i])$ \Comment*{\Oh{1} time per entry}
  }
\For(\Comment*[f]{Induce $D[\cdot,\ell]$ from $D[\cdot,\ell-1]$}){$\ell = 2$ to $n$}{\For(\Comment*[f]{Induce $D[i,\ell]$}){$i = 2$ to $i$}{\lIf{$\ell < i$}{$D[i, \ell] \gets \bot$
      }\lElse{$D[i, \ell] \gets \min(D[i-1,\ell], D[i-1,\ell-1]T[i])$ \label{lineDIEll}
      }
    }
  }
  \caption{Computing the lexicographically smallest subsequence $D[i,\ell]$ in $T[1..i]$ of length~$\ell$.}
  \label{algoLexSmallest}
\end{algorithm}

\begin{figure}[t]
	\begin{minipage}{0.35\linewidth}
  \centering{\includegraphics[scale=0.8]{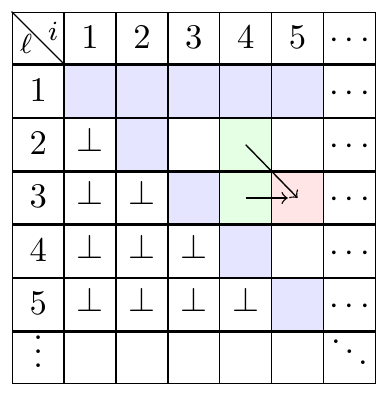}
  }\end{minipage}
	\hfill
	\begin{minipage}{0.55\linewidth}
  \caption{Sketch of the proof of \cref{lemLexSmallestDP}.
    We can fill the fields shaded in blue (the first row and the diagonal) in a precomputation step.
    Further, we know that entries left of the diagonal are all empty.
    A cell to the right of it (red) is based on its left-preceding and diagonal-preceding cell (green).
  }\label{figLexSmallestDP}
	\end{minipage}
\end{figure}

\begin{lemma}\label{lemLexSmallestDP}
  \Cref{algoLexSmallest} correctly computes $D[i,\ell]$, the lexicographically smallest subsequence of $T[1..i]$ with length~$\ell$.
\end{lemma}
\begin{proof}
  The proof is done by induction over the length~$\ell$ and the prefix~$T[1..i]$.
  We observe that $D[i,\ell] = \bot$ for $i < \ell$ and $D[i,i] = T[1..i]$
  since $T[1..i]$ has only one subsequence of length~$i$.
  Hence, for (a) $\ell = 1$ as well as for (b) $i \le \ell$, the claim holds.
  See \cref{figLexSmallestDP} for a sketch.

  Now assume that the claim holds for
  $D[i',\ell']$ with (a) $\ell' < \ell$ and all $i \in [1..n]$, as well as (b) $\ell' = \ell$ and all $i' \in [1..i-1]$.
  In what follows, we show that the claim also holds for $D[i,\ell]$ with $i > \ell > 1$.
  For that, let us assume that $T[1..i]$ has a subsequence~$L$ of length~$\ell$ with $L \prec D[i,\ell]$.

  If $L[\ell] \not= T[i]$, then $L$ is a subsequence of $T[1..i-1]$, and therefore $D[i-1,\ell] \preceq L$ according to the induction hypothesis.
  But $D[i,\ell] \preceq D[i-1,\ell]$, a contradiction.

  If $L[\ell] = T[i]$, then $L[1..\ell-1]$ is a subsequence of $T[1..i-1]$,
  and therefore $D[i-1,\ell-1] \preceq L[1..\ell-1]$ according to the induction hypothesis.
  But $D[i,\ell] \preceq D[i-1,\ell-1]T[i] \preceq L[1..\ell-1]T[i] = L$, a contradiction.
  Hence, $D[i,\ell]$ is the lexicographically smallest subsequence of $T[1..i]$ of length~$\ell$.
\end{proof}

Unfortunately, the lexicographically smallest subsequence of a given length is not a Lyndon word in general,
so this dynamic programming approach does not solve our problem finding the longest Lyndon subsequence.
In fact, if $T$ has a longest Lyndon subsequence of length~$\ell$, then there can be a lexicographically smaller subsequence of the same length. For instance, with $T = \texttt{aba}$, we have the longest Lyndon subsequence \texttt{ab}, while the lexicographically smallest length-2 subsequence is \texttt{aa}.

Analyzing the complexity bounds of \Cref{algoLexSmallest},
we need \Oh{n^2} space for storing the two-dimensional table $D[1..n,1..n]$.
Its initialization costs us $\Oh{n^2}$ time.
Line~\ref{lineDIEll} is executed \Oh{n^2} time.
There, we compute the lexicographical minimum of two subsequences.
If we evaluate this computation with naive character comparisons, for which we need to check \Oh{n} characters,
we pay \Oh{n^3} time in total, which is also the bottleneck of this algorithm.

\begin{lemma}\label{thmDPCubed}
  We can compute the lexicographically smallest substring of $T$ for each length~$\ell$
  online in \Oh{n^3} time with \Oh{n^2} space.
\end{lemma}

\subsection{Speeding Up String Comparisons}\label{secOrderMaintenance}

\newcommand*{\functionname}[1]{{\ensuremath{\renewcommand{\rmdefault}{ptm}\fontfamily{ppl}\selectfont\textrm{\textup{#1}}}}} \newcommand*{\fnPreceding}{\functionname{precedes}}
\newcommand*{\fnInsert}{\functionname{insert}}
\newcommand*{\fnDepth}{\functionname{depth}}
\newcommand*{\fnLevelanc}{\functionname{level-anc}}

Below, we improve the time bound of \cref{thmDPCubed} by representing each cell of $D[1..n,1..n]$ with a node in a trie, which supports the following methods:
\begin{itemize}
  \item $\fnInsert(v,c)$: adds a new leaf to a node~$v$ with an edge labeled with character~$c$, and returns a handle to the created leaf.
  \item $\fnPreceding(u, v)$: returns true if the string represented by the node~$u$ is lexicographically smaller than the string represented by the node~$v$.
\end{itemize}
Each cell of~$D$ stores a handle to its respective trie node.
The root node of the trie represents the empty string~$\bot$, and we associate $D[0,\ell] = \bot$ with the root node for all $\ell$.
A node representing $D[i-1,\ell-1]$ has a child representing $D[i,\ell]$ connected with an edge labeled with~$c$ if $D[i,\ell] = D[i-1,\ell-1]c$,
which is a concept similar to the LZ78 trie.
If $D[i,\ell] = D[i-1,\ell]$, then both strings are represented by the same trie node.
Since each node stores a constant number of words and an array storing its children,
the trie takes \Oh{n^2} space.

\paragraph*{\bf Insert.}
A particularity of our trie is that it stores the children of a node in the order of their creation, i.e.,
we always make a new leaf the last among its siblings.
This allows us to perform \fnInsert{} in constant time by representing the pointers to the children of a node by a plain dynamic array.
When working with the trie, we assure that we do not insert edges into the same node with the same character label (to prevent duplicates).

We add leaves to the trie as follows:
Suppose that we compute $D[i,\ell]$.
If we can copy $D[i-1,\ell]$ to $D[i,\ell]$ (Line~\ref{lineDIEll}),
we just copy the handle of~$D[i-1,\ell]$ pointing to its respective trie node to $D[i,\ell]$.
Otherwise, we create a new trie leaf,
where we create a new entry of~$D$ by selecting a new character ($\ell = 1$), or appending a character to one of the existing strings in~$D$.
We do not create duplicate edges since we prioritize copying to the creation of a new trie node:
For an entry~$D[i,\ell]$, we first default to the previous occurrence~$D[i-1,\ell]$,
and only create a new string~$D[i-1,\ell-1]T[i]$ if $D[i-1,\ell-1]T[i] \prec D[i-1,\ell]$.
$D[i-1,\ell-1]T[i]$ cannot have an occurrence represented in the trie.
To see that, we observe that $D$ obeys the invariants that
(a) $D[i,\ell] = \min_{j \in [1..i]} D[j,\ell]$ 
(where $\min$ selects the lexicographically minimal string)
and (b) all pairs of rows $D[\cdot,\ell]$ and $D[\cdot,\ell']$ with $\ell \neq \ell'$ have different entries.
Since \Cref{algoLexSmallest} fills the entries in $D[\cdot,\ell]$ in a lexicographically non-decreasing order for each length~$\ell$, we cannot create duplicates (otherwise, an earlier computed entry would be lexicographically smaller than a later computed entry having the same length).
The string comparison $D[i-1,\ell-1]T[i] \prec D[i-1,\ell]$ is done by calling $\fnPreceding{}$, which works as follows:

\paragraph*{\bf Precedes.}
We can implement the function \fnPreceding{} efficiently by augmenting our trie with the dynamic data structure of
\cite{cole05dynamic} supporting lowest common ancestor (LCA) queries in constant time
and the dynamic data structure of \cite{dietz91finding} supporting level ancestor queries $\fnLevelanc(u,d)$ returning the ancestor of a node~$u$ on depth~$d$ in amortized constant time.
Both data structures conform with our definition of \fnInsert{} that only supports the insertion of \emph{leaves}.
With these data structures, we can implement $\fnPreceding(u,v)$, by first computing the lowest ancestor~$w$ of $u$ and $v$,
selecting the children $u'$ and $v'$ of~$w$ on the paths downwards to $u$ and $v$, respectively, by two level ancestor queries~$\fnLevelanc(u,\fnDepth(w)+1)$ and~$\fnLevelanc(v,\fnDepth(w)+1)$,
and finally returning true if the label of the edge $(w,u')$ is smaller than of $(w,v')$.

We use \fnPreceding{} as follows for deciding whether $D[i-1,\ell-1]T[i] \prec D[i-1,\ell]$ holds:
Since we know that $D[i-1,\ell-1]$ and $D[i-1,\ell]$ are represented by nodes~$u$ and $v$ in the trie, respectively,
we first check whether $u$ is a child of $v$.
In that case, we only have to compare $T[i]$ with $D[i-1,\ell][\ell]$.
If not, then we know that $D[i-1,\ell-1]$ cannot be a prefix of $D[i-1,\ell]$, and $\fnPreceding(u,v)$ determines whether $D[i-1,\ell-1]$ or the $\ell-1$-th prefix of $D[i-1,\ell]$ is lexicographically smaller.

\begin{theorem}
  We can compute the table $D[1..n,1..n]$
  in \Oh{n^2} time using \Oh{n^2} words of space.
\end{theorem}

\newcommand*{\Stack}{\ensuremath{\mathsf{S}}}

\begin{figure}
  \centering{\includegraphics[scale=1.0]{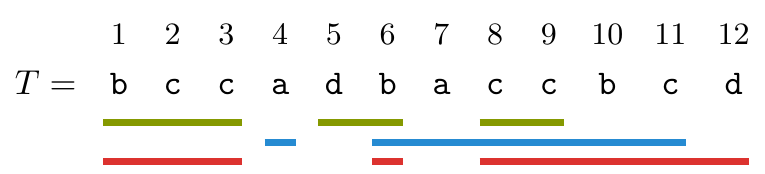}
  }\caption{Longest Lyndon subsequences of prefixes of a text~$T$.
    The $i$-th row of bars below $T$ depicts the selection of characters forming a Lyndon sequence.
    In particular, the $i$-th row corresponds to the longest subsequence of
    $T[1..9]$ for $i=1$ (green), $T[1..11]$ for $i=2$ (blue), and of $T[1..12]$ for $i=3$ (red).
    The first row (green) corresponds also to a longest Lyndon subsequence of $T[1..10]$ and $T[1..11]$.
    Extending the second Lyndon subsequence with $T[12]$ gives also a Lyndon subsequence, but is shorter than the third Lyndon subsequence (red).
    Having only the information of the Lyndon subsequences in $T[1..i]$ at hand seems not to give us a solution for $T[1..i+1]$.
  }
  \label{figCounterExampleDP}
\end{figure}

\newcommand*{\Top}{\textsf{top}}

\subsection{Most Competitive Subsequence}
If we want to find only the lexicographically smallest subsequence for a fixed length~$\ell$,
this problem is also called to
\emph{Find the Most Competitive Subsequence}\footnote{\url{https://leetcode.com/problems/find-the-most-competitive-subsequence/}}.
For that problem, there are linear-time solutions using a stack~\Stack{} storing the lexicographically smallest subsequence of length~$\ell$ for any prefix~$T[1..i]$ with $\ell \le i$.
The idea is to first fill \Stack{} with $1,\ldots,\ell$.
Let $\Top$ denote the top element of \Stack{}.
Given we are at text position $i > \ell$, we recursively pop \Top{} as long as $T[\Top] > T[i]$ and $n-i \ge (\ell - |\Stack|)$.
The latter condition ensures that when we are near the end of the text, we still have enough positions in \Stack{} to fill up \Stack{} with the remaining positions to obtain a sequence of $\ell$ text positions.
Finally, we put $T[i]$ on top of $\Stack{}$ if $|\Stack{}| < \ell$.
Since a text position gets inserted into \Stack{} and removed from \Stack{} at most once, the algorithm runs in linear time.
Consequently, if the whole text $T$ is given (i.e., not online),
this solution solves our problem in the same time and space bounds by running the algorithm for each $\ell$ separately.

\subsection{Lexicographically Smallest Common Subsequence}
Another variation is to ask for the lexicographically smallest subsequence of each distinct length that is common with two strings~$X$ and~$Y$.
Luckily, our ideas of \cref{secFirstApproachLexicographicallySmallest,secOrderMaintenance} can be straightforwardly translated.
For that, our matrix~$D$ becomes a cube $D_3[1..L,1..|X|,1..|Y|]$ with $L := \min(|X|,|Y|)$,
and we set
\[
D_3[\ell,x+1,y+1] = 
\min
\begin{cases}
D_3[\ell-1,x,y]X[x+1] \text{~if~} X[x+1] = Y[y+1], \\
D_3[\ell,x,y+1], \\
D_3[\ell,x+1,y],
\end{cases}
\]
with $D_3[0,\cdot,\cdot] = D_3[\ell,x,y] = \bot$ for all $\ell,x,y$ with $\LCS(X[1..x],Y[1..y]) < \ell$, where $\LCS$ denotes the length of a longest common subsequence of $X$ and $Y$.
This gives us an induction basis similar to the one used in the proof of \cref{lemLexSmallestDP}, such that
we can use its induction step analogously.
The table $D_3$ has $\Oh{n^3}$ cells, and filling each cell can be done in constant time by representing each cell as a pointer to a node in the trie data structure proposed in \cref{secOrderMaintenance}.
For that, we ensure that we never insert a subsequence of $D_3$ into the trie twice.
To see that, let $L \in \Sigma^+$ be a subsequence computed in $D_3$, and let $D_3[\ell,x,y] = L$ be the entry at which we called $\fnInsert$ to create a trie node for $L$ (for the first time).
Then $\ell = |L|$, and 
$X[1..x]$ and $Y[1..y]$ are the shortest prefixes of $X$ and $Y$, respectively, containing $L$ as a subsequence.
Since $D_3[\ell,x,y] = \min_{x' \in [1..x], y' \in [1..y]} D_3[\ell,x',y']$,
all other entries $D_3[\ell,x',y']=L$ satisfy $D_3[\ell,x'-1,y']=L$ or $D_3[\ell,x',y'-1]=L$,
so we copy the trie node handle representing $L$ instead of calling \fnInsert{} when filling out $D_3[\ell,x',y']$.

\begin{theorem}
	Given two strings $X, Y$ of length~$n$,
	we can compute the lexicographically smallest common subsequence for each length $\ell \in [1..n]$
	in $\Oh{n^3}$ time using $\Oh{n^3}$ space.
\end{theorem}

\section{Computing the Longest Lyndon Subsequence}\label{secLongestLyndon}

In the following, we want to compute the longest Lyndon subsequence of $T$.
See \cref{figCounterExampleDP} for examples of longest Lyndon subsequences.
Compared to the former introduced dynamic programming approach for the lexicographically smallest subsequences,
we follow the sketched solution for the most competitive subsequence using a stack, which here simulates a traversal
of the trie~$\tau$ storing all pre-Lyndon subsequences.
$\tau$ is a subgraph of the trie storing all subsequences, sharing the same root.
This subgraph is connected since, by definition, there is no string~$S$ such that $WS$ forms a pre-Lyndon word for a non-pre-Lyndon word~$W$
(otherwise, we could extend $WS$ to a Lyndon word, and so $W$, too).
We say that the \teigi{string label} of a node~$v$ is the string read from the edges on the path from root to~$v$.
We associate the label~$c$ of each edge of the trie with the leftmost possible position such that the string label~$V$ of~$v$
is associated with the sequence of text positions $i_1 < i_2 < \cdots < i_{|V|}$ and $T[i_1]T[i_2] \cdots T[i_{|V|}] = V$.

\subsection{Basic Trie Traversal}\label{secBasicTrieTraversal}
Problems already emerge when considering the construction of $\tau$
since there are texts like $T = 1 \cdots n$ for which $\tau$ has $\Ot{2^n}$ nodes.
Instead of building~$\tau$, we simulate a preorder traversal on it.
With simulation we mean that we enumerate the pre-Lyndon subsequences of~$T$ in lexicographic order.
For that, we maintain a stack~\Stack{} storing the text positions~$(i_1, \ldots, i_\ell)$ with $i_1 < \cdots < i_\ell$  associated with the path from the root to the node~$v$ we currently visit
i.e.,  $i_1, \ldots, i_\ell$ are the smallest positions with $T[i_1] \cdots T[i_\ell]$ being the string label of~$v$, 
which is a pre-Lyndon word.
When walking down, we select the next text position~$i_{\ell+1}$ such that $T[i_1] \cdots T[i_\ell]T[i_{\ell+1}]$ is a pre-Lyndon word.
If such a text position does not exist,
we backtrack by popping $i_\ell$ from~\Stack{}, and push the smallest text position~$i'_\ell > i_{\ell-1}$ with $T[i'_\ell] > T[i_\ell]$ onto~\Stack{} and recurse.
Finally, we check at each state of~\Stack{} storing the text positions~$(i_1, \ldots, i_\ell)$ whether $T[i_1] \cdots T[i_\ell]$ is a Lyndon word.
For that, we make use of the following facts:

\paragraph*{\bf Facts about Lyndon Words. }
A Lyndon word cannot have a \teigi{border}, that is, a non-empty proper prefix that is also a suffix of the string.
A \teigi{pre-Lyndon word} is a (not necessarily proper) prefix of a Lyndon word.
Given a string $S$ of length $n$, an integer $p \in [1..n]$ is a \teigi{period} of $S$
if $S[i] = S[i+p]$ for all $i \in [1..n-p]$.
The length of a string is always one of its periods.
We use the following facts:
\begin{enumerate}[label=(Fact~\arabic*), ref=\arabic*,leftmargin=*]
  \item Only the length $|S|$ is the period of a Lyndon word~$S$.
  \item The prefix $S[1..|p|]$ of a pre-Lyndon word~$S$ with period~$p$ is a Lyndon word.
        In particular, a pre-Lyndon word~$S$ with period~$|S|$ is a Lyndon word.
        \label{factPreLyndonPeriod}
  \item Given a pre-Lyndon word~$S$ with period~$p$ and a character~$c \in \Sigma$, then
        \begin{itemize}
          \item $Sc$ is a pre-Lyndon word of the same period if and only if $S[|S|-p+1] = c$
          \item $Sc$ is a Lyndon word if and only if $S[|S|-p+1] < c$.
                In particular, if $S$ is a Lyndon word, then $Sc$ is a Lyndon word if and only if $S[1]$ is smaller than $c$.
        \end{itemize}
        \label{factLyndonExtension}
\end{enumerate}

\paragraph*{\bf Checking pre-Lyndon Words. }
Now suppose that our stack~\Stack{} stores the text positions~$(i_1, \ldots, i_\ell)$.
To check whether $T[i_1] \cdots T[i_\ell] c$ for a character $c \in \Sigma$ is a pre-Lyndon word or whether it is a Lyndon word,
we augment each position~$i_j$ stored in~\Stack{} with the period of $T[i_1] \cdots T[i_j]$, for $j \in [1..\ell]$, such that
we can make use of Fact~\ref{factLyndonExtension} to compute the period and check whether $T[i_1] \cdots T[i_j] c$ is a pre-Lyndon word, both in constant time, for $c \in \Sigma$.

\paragraph*{\bf Trie Navigation. }
To find the next text position~$i_{\ell+1}$, we may need to scan \Oh{n} characters in the text, and hence need \Oh{n} time for walking down from a node to one of its children.
If we restrict the alphabet to be integer,
we can augment each text position~$i$ to store the smallest text position~$i_c$ with $i < i_c$ for each character~$c \in \Sigma$ such that
we can visit the trie nodes in constant time per node during our preorder traversal.

This gives already an algorithm that computes the longest Lyndon subsequence with \Oh{n \sigma} space and time
linear to the number of nodes in $\tau$.
However, since the number of number can be exponential in the text length, 
we present ways to omit nodes that do not lead to the solution.
Our aim is to find a rule to judge whether a trie node contributes to the longest Lyndon subsequence to leave certain subtrees of the trie unexplored.
For that, we use the following property:

\begin{lemma}\label{lemTakeLongerLyndon}
  Given a Lyndon word~$V$ and
  two strings~$U$ and $W$
  such that $UW$ is a Lyndon word,
  $V \prec U$, and $|V| \ge |U|$, then
  $VW$ is also a Lyndon word with $VW \prec UW$.
\end{lemma}
\begin{proof}
  Since $V \prec U$ and $V$ is not a prefix of~$U$, $U \succ VW$.
  In what follows, we show that $S \succ VW$ for every proper suffix~$S$ of $VW$.
  \begin{itemize}
    \item If $S$ is a suffix of~$W$, then
          $S \succeq UW \succeq U \succ VW$
          because $S$ is a suffix of the Lyndon word~$UW$.
    \item Otherwise, ($|S| > |W|$), $S$ is of the form $V'W$ for a proper suffix~$V'$ of~$V$.
          Since $V$ is a Lyndon word, $V' \succ V$, and $V'$ is not a prefix of~$V$ (Lyndon words are border-free).
          Hence, $V' W \succeq V' \succ VW$.
  \end{itemize}
\end{proof}
Note that $U$ in \cref{lemTakeLongerLyndon} is a pre-Lyndon word since it is the prefix of the Lyndon word~$UW$.

\newcommand*{\Len}{\ensuremath{\mathsf{L}}}

Our algorithmic idea is as follows:
We maintain an array~$\Len[1..n]$, where
$\Len[\ell]$ is the smallest text position~$i$ 
such that our traversal has already explored a length-$\ell$ Lyndon subsequence of $T[1..i]$.
We initialize the entries of $\Len$ with $\infty$ at the beginning.
Now, whenever we visit a node~$u$ whose string label is a pre-Lyndon subsequence~$U = T[i_1] \cdots T[i_\ell]$ with
$\Len[\ell] \le i_\ell$, then we do not explore the children of~$u$.
In this case, we call $u$ \teigi{irrelevant}.
By skipping the subtree rooted at~$u$, we do not omit the solution due to \cref{lemTakeLongerLyndon}:
When $\Len[\ell] \le i_\ell$, then there is a Lyndon subsequence~$V$ of $T[1..i_\ell]$ with $V \prec U$ (since we traverse the trie in lexicographically order)
and $|V|=|U|$.
Given there is a Lyndon subsequence $UW$ of $T$, then we have already found $VW$ earlier, which is also a Lyndon subsequence of $T$ with $|VW| = |UW|$.

Next, we analyze the complexity of this algorithm, and propose an improved version.
For that, we say that a string is \teigi{immature} if it is pre-Lyndon but not Lyndon.
We also consider a subtree rooted at a node~$u$ as pruned if $u$ is irrelevant, i.e., the algorithm does not explore this subtree. 
Consequently, irrelevant nodes are leaves in the pruned subtree, 
but not all leaves are irrelevant (consider a Lyndon subsequence using the last text position~$T[n]$).
Further, we call a node Lyndon or immature if its string label is Lyndon or immature, respectively. 
(All nodes in the trie are either Lyndon or immature.)

\paragraph*{\bf Time Complexity. }
Suppose that we have the text positions~$(i_1, \ldots, i_\ell)$ on~\Stack{} such that $U := T[i_1] \cdots T[i_\ell]$ is a Lyndon word.
If $\Len[\ell] > i_\ell$, then we lower $\Len[\ell] \gets i_\ell$.
We can lower an individual entry of~$\Len$ at most $n$ times, or at most $n^2$ times in total for all entries.
If a visited node is Lyndon,
we only explore its subtree if we were able to lower an entry of $\Len$.
Hence, we visit at most $n^2$ Lyndon nodes that trigger a decrease of the values in $\Len$.
While each node can have at most $\sigma$ children,
at most one child can be immature.
Since the depth of the trie is at most $n$,
we therefore visit $\Oh{n\sigma}$ nodes
between two updates of $\Len$.
These nodes are leaves (of the pruned trie) or immature nodes.
Thus, we traverse $\Oh{n^3 \sigma}$ nodes in total.

\begin{theorem}
  We can compute the longest Lyndon subsequence
  of a string of length~$n$ 
  in $\Oh{n^3 \sigma}$ time using $\Oh{n \sigma}$ words of space.
\end{theorem}

\subsection{Improving Time Bounds}

We further improve the time bounds by avoiding visiting irrelevant nodes
due to the following observation:
First, we observe that the number of {\em relevant} (i.e., non-irrelevant) nodes that are Lyndon is $\Oh{n^2}$. 
Since all nodes have a depth of at most $n$, the total number of relevant nodes in $\Oh{n^3}$.
Suppose we are at a node~$u$, and~\Stack{} stores the positions $(i_1, \ldots, i_\ell)$ such that $T[i_1] \cdots T[i_\ell]$ is the string label of~$u$. Let $p$ denote the smallest period of $T[i_1]\cdots T[i_\ell]$.
Then we do not want to consider all $\sigma$ children of~$u$, but only those whose edges to~$u$ have a label~$c \geq T[i_{\ell-p+1}]$ such that
$c$ occurs in $T[i_\ell+1..\Len[\ell+1]-1]$ (otherwise, there is already a Lyndon subsequence of length~$\ell+1$ lexicographically smaller than $T[i_1]\cdots T[i_\ell] c$).
In the context of our preorder traversal,
each such child can be found iteratively using range successor queries:
starting from $b = T[i_{\ell-p+1}]-1$,
we want to find the \emph{lexicographically smallest} character~$c > b$
such that $c$ occurs in $T[i_\ell+1..\Len[\ell+1]-1]$.
In particular, we want to find the leftmost such occurrence.
A data structure for finding $c$ in this interval is the wavelet tree~\cite{grossi03wavelet}
returning the position of the \emph{leftmost} such $c$ (if it exists) in $\Oh{\lg \sigma}$ time.
In particular, we can use the wavelet tree instead of the $\Oh{n \sigma}$ pointers to the subsequent characters to arrive at \Oh{n} words of space.
Finally, we do not want to query the wavelet tree each time, but only whenever we are sure that it will lead us to a relevant Lyndon node.
For that, we build a range maximum query (RMQ) data structure on the characters of the text~$T$ in a preprocessing step.
The RMQ data structure of~\cite{bender05lowest} can be built in $\Oh{n}$ time; it answers queries in constant time.
Now, in the context of the above traversal where we are at a node~$u$ with \Stack{} storing $(i_1, \ldots, i_\ell)$,
we query this RMQ data structure for the largest character~$c$ in $T[i_\ell+1..\Len[\ell+1]-1]$
and check whether the sequence $S := T[i_1] \cdots T[i_\ell]c$ forms a (pre-)Lyndon word.
\begin{itemize}
  \item If $S$ is not pre-Lyndon, i.e., $T[i_{\ell-p+1}] > c$ for $p$ being the smallest period of $T[i_1] \cdots T[i_\ell]$,
  we are sure that the children of $u$ cannot lead to Lyndon subsequences~\cite[Prop.~1.5]{duval83lyndon}.
  \item If $S$ is immature, i.e., $T[i_{\ell-p+1}] = c$, 
    $u$ has exactly one child, and this child's string label is $S$.
  Hence, we do not need to query for other Lyndon children.
  \item Finally, if $S$ is Lyndon, i.e., $T[i_{\ell-p+1}] < c$,
  we know that there is at least one child of $u$ that will trigger an update in $\Len$ and thus is a relevant node.
\end{itemize}
This observation allows us to find all relevant children of $u$ (including the single immature child, if any) 
by iteratively conducting $\Oh{k}$ range successor queries, where $k$ is the number of children of $u$ that are relevant Lyndon nodes.
Thus, if we condition the execution of the aforementioned wavelet tree query with an RMQ query result on the same range,
the total number of wavelet tree queries can be bounded by $\Oh{n^2}$.
This gives us $\Oh{n^3 + n^2 \lg \sigma} = \Oh{n^3}$ time for $\sigma = \Oh{n}$ (which can be achieved by an $\Oh{n\log n}$ time re-enumeration of the alphabet in a preliminary step).

\begin{theorem}
  We can compute the longest Lyndon subsequence
  of a string of length~$n$ 
  in \Oh{n^3} time using \Oh{n} words of space.
\end{theorem}

In particular, the algorithm computes the lexicographically smallest one among all longest Lyndon subsequences:
Assume that this subsequence~$L$ is not computed, then we did not explore the subtree of the original trie~$\tau$ (before pruning) containing the node with string label~$L$.
Further, assume that this subtree is rooted at an irrelevant node~$u$ whose string label is the pre-Lyndon subsequence~$U$.
Then $U$ is a prefix of $L$, 
and because $u$ is irrelevant (i.e.,  we have not explored $u$'s children), there is a node~$v$ whose string label is a Lyndon word~$V$ with $V \prec U$ and $|V| = |U|$.
In particular, the edge of~$v$ to $v$'s parent is associated with a text position equal to or smaller than the associated text position of the edge between~$u$ and $u$'s parent.
Hence, we can extend $V$ to the Lyndon subsequence $V L[|U|+1]..]$ being lexicographically smaller than $L$, a contradiction.

\subsection{Online Computation}\label{secLyndonOnlineComputation}
If we allow increasing the space usage in order to maintain the trie data structure introduced in \cref{secOrderMaintenance}, 
we can modify our \Oh{n^3 \sigma}-time algorithm of \cref{secBasicTrieTraversal} to perform the computation online, i.e., with $T$ given as a text stream.
For that, we explicitly represent the visited nodes of the trie $\tau$ with an explicit trie data structure~$\tau'$
such that we can create pointers to the nodes.
(In other words, $\tau'$ is a lazy representation of $\tau$.)
The problem is that we can no longer perform the traversal in lexicographic order,
but instead keep multiple fingers in the trie $\tau'$ constructed up so far, and use these fingers to advance the trie traversal in text order.

With a different traversal order, we need an updated definition of $\Len[1..n]$: Now, while the algorithm processes $T[i]$,
the entry $\Len[\ell]$ stores the lexicographically smallest length-$\ell$ Lyndon subsequence of $T[1..i]$ (represented by a pointer to the corresponding node of $\tau'$).
Further, we maintain $\sigma$ lists storing pointers to nodes of $\tau'$.
Initially, $\tau'$ consists only of the root node, and each list stores only the root node.
Whenever we read a new character~$T[i]$ from the text stream,
for each node~$v$ of the \mbox{$T[i]$-th} list, we add a leaf~$\lambda$ connected to $v$ by an edge with label~$T[i]$. Our algorithm adheres to the invariant that $\lambda$'s string label $S$ is a pre-Lyndon word 
so that $\tau'$ is always a subtree of $\tau$.
 If $S$ is a Lyndon word satisfying $S\prec \Len[|S|]$ (which can be tested using the data structure of \cref{secOrderMaintenance}), we further set $\Len[|S|]:=S$. 
This completes the process of updating $\Len[1..n]$.
Next, we clear the $T[i]$-th list and iterate again over the newly created leaves. For each such leaf $\lambda$ with label $S$, we check whether $\lambda$ is relevant, i.e., whether $S \preceq \Len[|S|]$.
If $\lambda$ turns out irrelevant, we are done with processing it.
Otherwise, we put $\lambda$ into the $c$-th list for each character $c \in \Sigma$ such that $Sc$ is a pre-Lyndon word.
By doing so, we effectively create new events that trigger a call-back to the point where we stopped the trie traversal.

Overall, we generate exactly the nodes visited by the algorithm of \cref{secBasicTrieTraversal}.
In particular, there are \Oh{n^3} relevant nodes, and for each such node, we issue \Oh{\sigma} events. The operations of \cref{secOrderMaintenance}
take constant amortized time, so the overall time and space complexity of the algorithm are \Oh{n^3\sigma}.

\begin{theorem}
  We can compute the longest Lyndon subsequence online in \Oh{n^3 \sigma} time using \Oh{n^3 \sigma} space.
\end{theorem}

\section{Longest Common Lyndon Subsequence}\label{secLongestCommonLyndon}

Given two strings $X$ and $Y$, we want to compute the longest common subsequence of $X$ and $Y$ that is Lyndon.
For that, we can extend our algorithm finding the longest Lyndon subsequence of a single string as follows.
First, we explore in depth-first order the trie of all \emph{common} pre-Lyndon subsequences of $X$ and $Y$.
A node is represented by a pair of positions $(x, y)$ such that,
given the path from the root to a node~$v$ of depth~$\ell$ visits the nodes
$(x_1, y_1), \ldots, (x_\ell, y_\ell)$ with $L = X[x_1] \cdots X[x_\ell] = Y[y_1] \cdots Y[y_\ell]$ being a pre-Lyndon word,
$L$ is neither a subsequence of $X[1..x_\ell-1]$ nor of $Y[1..y_\ell-1]$,
i.e., $x_\ell$ and $y_\ell$ are the leftmost such positions.
The depth-first search works like an exhaustive search in that it tries to extend
$L$ with each possible character in $\Sigma$ having an occurrence in both remaining suffixes
$X[x_{\ell}+1..]$ and $Y[y_{\ell}+1..]$, and then,
after having explored the subtree rooted at~$v$,
visits its lexicographically succeeding sibling nodes (and descends into their subtrees)
by checking whether $L[1..|L|-1]$ can be extended with a character $c > L[|L|]$ appearing in
both suffixes $X[x_{\ell-1}+1..]$ and $Y[y_{\ell-1}+1..]$.

The algorithm uses again the array $\Len$ to check whether we have already found a lexicographically smaller Lyndon subsequence with equal or smaller ending positions in $X$ and $Y$ than the currently constructed pre-Lyndon subsequence.
For that, $\Len[\ell]$ stores not only one position,
but a list of positions $(x,y)$ such that $X[1..x]$ and $Y[1..y]$ have a \emph{common} Lyndon subsequence of length~$\ell$.
Although there can be $n^2$ such pairs of positions, we only store those that are pairwise non-dominated.
A pair of positions $(x_1,y_1)$ is called \teigi{dominated} by a pair $(x_2,y_2) \neq (x_1,y_1)$ if $x_2 \le x_1$ and $y_2 \le y_1$.
A set storing pairs in $[1..n] \times [1..n]$ can have at most $n$ elements that are pairwise non-dominated,
and hence $|\Len[\ell]| \le n$.

At the beginning, all lists of $\Len$ are empty.
Suppose that we visit a node $v$ with pair $(x_\ell,y_\ell)$ representing a common Lyndon subsequence of length~$\ell$.
Then we query whether $\Len[\ell]$ has a pair dominating $(x_\ell,y_\ell)$.
In that case, we can skip $v$ and its subtree.
Otherwise, we insert $(x_\ell,y_\ell)$ and remove pairs in $\Len[\ell]$ that are dominated by $(x_\ell,y_\ell)$.
Such an insertion can happen at most $n^2$ times.
Since $\Len[1..n]$ maintains $n$ lists, we can update $\Len$ at most $n^3$ times in total.
Checking for domination and insertion into $\Len$ takes \Oh{n} time.
The former can be accelerated to constant time by representing $\Len[\ell]$ as an array~$R_\ell$
storing in $R_\ell[i]$ the value $y$ of the tuple $(x,y) \in \Len[\ell]$ with $x \le i$ and the lowest possible $y$, for each $i \in [1..n]$.
Then a pair $(x,y) \not\in \Len[\ell]$ is dominated if and only if $R_\ell[x] \le y$.

\begin{example}
  For $n = 10$, let $\Len_\ell = [(3,9), (5,4), (8,2)]$.
  Then all elements in $\Len_\ell$ are pairwise non-dominated, and
  $R_\ell = [\infty,\infty,9,9,4,4,4,2,2,2]$.
  Inserting $(3,2)$ would remove all elements of $\Len_\ell$, and update all entries of $R_\ell$.
  Alternatively, inserting $(7,3)$ would only involve updating $R_\ell[7] \gets 3$;
  since the subsequent entry $R_\ell[8] = 2$ is less than $R_\ell[7]$, no subsequent entries need to be updated.
\end{example}

An update in $\Len[\ell]$ involves changing \Oh{n} entries of $R_\ell$,
but that cost is dwarfed by the cost for finding the next common Lyndon subsequence that updates $\Len$.
Such a subsequence can be found while visiting $\Oh{n \sigma}$ irrelevant nodes
during a naive depth-first search (cf.\ the solution of \cref{secFirstApproachLexicographicallySmallest} computing the longest Lyndon sequence of a single string).
Hence, the total time is $\Oh{n^4 \sigma}$.

\begin{theorem}
  We can compute the longest common Lyndon subsequence
  of a string of length~$n$ 
  in $\Oh{n^4 \sigma}$ time using $\Oh{n^3}$ words of space.
\end{theorem}

\section{Open Problems}
Since we shed light on the computation of the longest (common) Lyndon subsequence for the very first time,
we are unaware of the optimality of our solutions.
It would be interesting to find non-trivial lower bounds that would justify our rather large time and space complexities.

\subsubsection*{Acknowledgments}
This work was supported by JSPS KAKENHI Grant Numbers
JP20H04141 (HB),
JP19K20213 (TI),
JP21K17701 and JP21H05847 (DK).
TK was supported by NSF 1652303, 1909046, and HDR TRIPODS 1934846 grants, and an Alfred P. Sloan Fellowship.


\begin{thebibliography}{10}
\providecommand{\url}[1]{\texttt{#1}}
\providecommand{\urlprefix}{URL }
\providecommand{\doi}[1]{https://doi.org/#1}

\bibitem{BIINTT17}
Bannai, H., I, T., Inenaga, S., Nakashima, Y., Takeda, M., Tsuruta, K.: The
  "runs" theorem. {SIAM} J. Comput.  \textbf{46}(5),  1501--1514 (2017).
  \doi{10.1137/15m1011032}

\bibitem{GBRobinson1938}
de~Beauregard~Robinson, G.: On the representations of the symmetric group.
  American Journal of Mathematics  \textbf{60}(3),  745--760 (1938).
  \doi{10.2307/2371609}

\bibitem{bender05lowest}
Bender, M.A., Farach{-}Colton, M., Pemmasani, G., Skiena, S., Sumazin, P.:
  Lowest common ancestors in trees and directed acyclic graphs. J. Algorithms
  \textbf{57}(2),  75--94 (2005). \doi{10.1016/j.jalgor.2005.08.001}

\bibitem{DBLP:journals/siamdm/BiedlBCLMNS19}
Biedl, T.C., Biniaz, A., Cummings, R., Lubiw, A., Manea, F., Nowotka, D.,
  Shallit, J.O.: Rollercoasters: Long sequences without short runs. {SIAM} J.
  Discret. Math.  \textbf{33}(2),  845--861 (2019). \doi{10.1137/18M1192226}

\bibitem{DBLP:conf/gd/BiedlCDJL17}
Biedl, T.C., Chan, T.M., Derka, M., Jain, K., Lubiw, A.: Improved bounds for
  drawing trees on fixed points with {L}-shaped edges. In: Frati, F., Ma, K.
  (eds.) 25th International Symposium on Graph Drawing and Network
  Visualization, {GD} 2017. LNCS, vol. 10692, pp. 305--317. Springer (2017).
  \doi{10.1007/978-3-319-73915-1_24}

\bibitem{chen58lyndon}
Chen, K.T., Fox, R.H., Lyndon, R.C.: Free differential calculus, {IV}. {T}he
  quotient groups of the lower central series. Ann. Math. pp. 81--95 (1958).
  \doi{10.1007/978-1-4612-2096-1_10}

\bibitem{chowdhury14computing}
Chowdhury, S.R., Hasan, M.M., Iqbal, S., Rahman, M.S.: Computing a longest
  common palindromic subsequence. Fundam. Inform.  \textbf{129}(4),  329--340
  (2014). \doi{10.3233/fi-2014-974}

\bibitem{cole05dynamic}
Cole, R., Hariharan, R.: Dynamic {LCA} queries on trees. {SIAM} J. Comput.
  \textbf{34}(4),  894--923 (2005). \doi{10.1137/s0097539700370539}

\bibitem{dietz91finding}
Dietz, P.F.: Finding level-ancestors in dynamic trees. In: Dehne, F.K.H.A.,
  Sack, J., Santoro, N. (eds.) 2nd Workshop on Algorithms and Data Structures,
  {WADS} 1991. LLNCS, vol.~519, pp. 32--40. Springer (1991).
  \doi{10.1007/BFb0028247}

\bibitem{duval83lyndon}
Duval, J.: Factorizing words over an ordered alphabet. J. Algorithms
  \textbf{4}(4),  363--381 (1983). \doi{10.1016/0196-6774(83)90017-2}

\bibitem{elmasry10longest}
Elmasry, A.: The longest almost-increasing subsequence. Inf. Process. Lett.
  \textbf{110}(16),  655--658 (2010). \doi{10.1016/j.ipl.2010.05.022}

\bibitem{DBLP:conf/spire/FujitaNIBT21}
Fujita, K., Nakashima, Y., Inenaga, S., Bannai, H., Takeda, M.: Longest common
  rollercoasters. In: Lecroq, T., Touzet, H. (eds.) 28th International
  Symposium on String Processing and Information Retrieval, {SPIRE} 2021.
  LLNCS, vol. 12944, pp. 21--32. Springer (2021).
  \doi{10.1007/978-3-030-86692-1_3}

\bibitem{DBLP:conf/stacs/GawrychowskiMS19}
Gawrychowski, P., Manea, F., Serafin, R.: Fast and longest rollercoasters. In:
  Niedermeier, R., Paul, C. (eds.) 36th International Symposium on Theoretical
  Aspects of Computer Science, {STACS} 2019. LIPIcs, vol.~126, pp. 30:1--30:17.
  Schloss Dagstuhl--Leibniz-Zentrum f{\"{u}}r Informatik (2019).
  \doi{10.4230/LIPIcs.STACS.2019.30}

\bibitem{grossi03wavelet}
Grossi, R., Gupta, A., Vitter, J.S.: High-order entropy-compressed text
  indexes. In: 14th Annual {ACM-SIAM} Symposium on Discrete Algorithms, {SODA}
  2003. pp. 841--850. {ACM/SIAM} (2003),
  \url{http://dl.acm.org/citation.cfm?id=644108.644250}

\bibitem{he18longest}
He, X., Xu, Y.: The longest commonly positioned increasing subsequences
  problem. J. Comb. Optim.  \textbf{35}(2),  331--340 (2018).
  \doi{10.1007/s10878-017-0170-9}

\bibitem{hirakawa21counting}
Hirakawa, R., Nakashima, Y., Inenaga, S., Takeda, M.: Counting lyndon
  subsequences. In: Holub, J., Zd{\'{a}}rek, J. (eds.) Prague Stringology
  Conferene. pp. 53--60. Czech Technical University in Prague (2021),
  \url{http://www.stringology.org/event/2021/p05.html}

\bibitem{hirschberg77algorithms}
Hirschberg, D.S.: Algorithms for the longest common subsequence problem. J.
  {ACM}  \textbf{24}(4),  664--675 (1977). \doi{10.1145/322033.322044}

\bibitem{inenaga18hardness}
Inenaga, S., Hyyr{\"{o}}, H.: A hardness result and new algorithm for the
  longest common palindromic subsequence problem. Inf. Process. Lett.
  \textbf{129},  11--15 (2018). \doi{10.1016/j.ipl.2017.08.006}

\bibitem{kiyomi21longest}
Kiyomi, M., Horiyama, T., Otachi, Y.: Longest common subsequence in sublinear
  space. Inf. Process. Lett.  \textbf{168},  106084 (2021).
  \doi{10.1016/j.ipl.2020.106084}

\bibitem{Knuth70}
Knuth, D.: Permutations, matrices, and generalized {Y}oung tableaux. Pac. J.
  Math.  \textbf{34},  709--727 (1970). \doi{10.2140/pjm.1970.34.709}

\bibitem{kosche21absent}
Kosche, M., Ko{\ss}, T., Manea, F., Siemer, S.: Absent subsequences in words.
  In: Bell, P.C., Totzke, P., Potapov, I. (eds.) 15th International Conference
  on Reachability Problems, {RP} 2021. LNCS, vol. 13035, pp. 115--131. Springer
  (2021). \doi{10.1007/978-3-030-89716-1_8}

\bibitem{kosowski04efficient}
Kosowski, A.: An efficient algorithm for the longest tandem scattered
  subsequence problem. In: Apostolico, A., Melucci, M. (eds.) 11th
  International Conference on String Processing and Information Retrieval.
  LNCS, vol.~3246, pp. 93--100. Springer (2004).
  \doi{10.1007/978-3-540-30213-1_13}

\bibitem{kutz11faster}
Kutz, M., Brodal, G.S., Kaligosi, K., Katriel, I.: Faster algorithms for
  computing longest common increasing subsequences. J. Discrete Algorithms
  \textbf{9}(4),  314--325 (2011). \doi{10.1016/j.jda.2011.03.013}

\bibitem{lyndon54}
Lyndon, R.C.: On {Burnside}'s problem. Trans. Am. Math. Soc.  \textbf{77}(2),
  202--215 (1954). \doi{10.1090/s0002-9947-1954-0064049-x}

\bibitem{Schensted61}
Schensted, C.: Longest increasing and decreasing subsequences. Can. J. Math.
  \textbf{13},  179--191 (1961). \doi{10.4153/CJM-1961-015-3}

\bibitem{ta21computing}
Ta, T.T., Shieh, Y., Lu, C.L.: Computing a longest common almost-increasing
  subsequence of two sequences. Theor. Comput. Sci.  \textbf{854},  44--51
  (2021). \doi{10.1016/j.tcs.2020.11.035}

\bibitem{wagner74correction}
Wagner, R.A., Fischer, M.J.: The string-to-string correction problem. J. {ACM}
  \textbf{21}(1),  168--173 (1974). \doi{10.1145/321796.321811}

\end{thebibliography}
\end{document}